\documentclass[conference,letterpaper]{IEEEtran}

\onecolumn

\usepackage[utf8]{inputenc} 
\usepackage[T1]{fontenc}
\usepackage{url}
\usepackage{ifthen}
\usepackage{cite}
\usepackage[cmex10]{amsmath} % Use the [cmex10] option to ensure complicance
\usepackage{bm,bbm,amssymb,mathrsfs,amsthm,mathtools,enumitem}

\usepackage{xcolor}

\newtheorem{theorem}{Theorem}

\theoremstyle{definition}
\newtheorem{definition}{Definition}

\newtheorem{lemma}{Lemma}
\newtheorem{corollary}{Corollary}

\newtheorem{assumption}{Assumption}

\newtheorem{remark}{Remark}

\newcommand{\R}{{\mathbb R}}

\newcommand{\E}{{\mathbb E}}

\newcommand{\Py}{{\mathbb P}}

\newcommand{\vct}[1]{\bm{#1}}

\newcommand{\rv}[1]{{{#1}}}

\DeclareMathOperator*{\argmax}{arg\,max}
\DeclareMathOperator*{\argmin}{arg\,min}

\begin{document}
\title{Active Hypothesis Testing: Beyond Chernoff-Stein} 

% %%% Single author, or several authors with same affiliation:
% \author{%
%   \IEEEauthorblockN{Stefan M.~Moser}
%   \IEEEauthorblockA{ETH Zürich\\
%                     ISI (D-ITET)\\
%                     CH-8092 Zürich, Switzerland\\
%                     Email: moser@isi.ee.ethz.ch}
% }

%%% Several authors with up to three affiliations:
\author{%
  \IEEEauthorblockN{Dhruva Kartik, Ashutosh Nayyar and Urbashi Mitra}
  \IEEEauthorblockA{Ming Hsieh Department of Electrical Engineering\\
                    University of Southern California, Los Angeles, CA, USA\\ 
                    %3740 McClintock Ave, Los Angeles CA 90089\\
                    Email: \{mokhasun, ashutosh.nayyar, ubli\}@usc.edu}
}

%%% Many authors with many affiliations:
% \author{%
%   \IEEEauthorblockN{Albus Dumbledore\IEEEauthorrefmark{1},
%                     Olympe Maxime\IEEEauthorrefmark{2},
%                     Stefan M.~Moser\IEEEauthorrefmark{3}\IEEEauthorrefmark{4},
%                     and Harry Potter\IEEEauthorrefmark{1}}
%   \IEEEauthorblockA{\IEEEauthorrefmark{1}%
%                     Hogwarts School of Witchcraft and Wizardry,
%                     1714 Hogsmeade, Scotland,
%                     \{dumbledore, potter\}@hogwarts.edu}
%   \IEEEauthorblockA{\IEEEauthorrefmark{2}%
%                     Beauxbatons Academy of Magic,
%                     1290 Pyrénées, France,
%                     maxime@beauxbatons.edu}
%   \IEEEauthorblockA{\IEEEauthorrefmark{3}%
%                     ETH Zürich, ISI (D-ITET), ETH Zentrum, 
%                     CH-8092 Zürich, Switzerland,
%                     moser@isi.ee.ethz.ch}
%   \IEEEauthorblockA{\IEEEauthorrefmark{4}%
%                     National Chiao Tung University (NCTU), 
%                     Hsinchu, Taiwan,
%                     moser@isi.ee.ethz.ch}
% }

\maketitle

%%%%%%
%% Abstract: 
%% If your paper is eligible for the student paper award, please add
%% the comment "THIS PAPER IS ELIGIBLE FOR THE STUDENT PAPER
%% AWARD." as a first line in the abstract. 
%% For the final version of the accepted paper, please do not forget
%% to remove this comment!
%%
\begin{abstract}
   An active hypothesis testing problem is formulated. In this problem, the agent can perform a fixed number of experiments and then decide on one of the hypotheses. The agent is also allowed to declare its experiments inconclusive if needed. The objective is to minimize the probability of making an incorrect inference (misclassification probability) while ensuring that the true hypothesis is declared conclusively with moderately high probability. For this problem, lower and upper bounds on the optimal misclassification probability are derived and these bounds are shown to be asymptotically tight. In the analysis, a sub-problem, which can be viewed as a generalization of the Chernoff-Stein lemma, is formulated and analyzed. A heuristic approach to strategy design is proposed and its relationship with existing heuristic strategies is discussed.
\end{abstract}

%% The paper must be self-contained. However, if you are referring to
%% a full version for checking certain proofs, please provide the
%% publically accessible location below.  If the paper is completely
%% self-contained, you can remove the following line from your
%% submission.
%\textit{A full version is accessible at:}
%\url{https://bit.ly/2DYsaa7} 

\section{Introduction}
We frequently encounter scenarios wherein we would like to deduce whether one of several hypotheses is true by gathering data or evidence. This problem is referred to as multi-hypothesis testing. If we have access to multiple candidate experiments or data sources, we can adaptively select more informative experiments to infer the true hypothesis. This leads to a joint control and inference problem commonly referred to as \emph{active} hypothesis testing. There are numerous ways of formulating this problem and the precise mathematical formulation depends on the target application.

In this paper, we consider a scenario in which there is an agent that can perform a \emph{fixed} number of experiments. Subsequently, the agent can decide on one of the hypotheses using the collected data. The agent is also allowed to declare the experiments \emph{inconclusive} if needed. The objective is to minimize the probability of making an
\emph{incorrect} inference (misclassification probability) while ensuring
that the true hypothesis is declared conclusively with moderately
high probability. This formulation is of particular interest when the agent is time-constrained and the penalty for making an incorrect inference is significantly higher than the penalty for making no decision. In such cases, it is reasonable for the agent to abstain from drawing conclusions unless there is strong evidence supporting one of the hypotheses. 

{For example, consider a decentralized system in which an agent needs to perform experiments and convey its results to another decision-maker (such as a fusion center). Due to communication constraints, the agent can only communicate its \emph{estimate} of the hypothesis or remain silent. The agent incurs heavy penalty for transmitting an incorrect hypothesis. However, the agent is also constrained to transmit the true hypothesis with moderately high probability. Thus, we would like to design an experiment selection strategy and an inference (transmission) strategy for the agent which minimize the misclassification probability while ensuring that the correct estimate is transmitted with sufficiently high probability.}

Our contributions in this paper can be summarized as follows. We find lower and upper bounds on the optimal misclassification probabilities in our constrained problem. These bounds are asymptotically tight under some mild assumptions. In our analysis, we formulate a \emph{sub-problem} and use the results from the sub-problem to solve our original problem. This sub-problem can be viewed as a generalization of the Chernoff-Stein lemma \cite{cover2012elements} to a setting with multiple hypotheses and multiple experiments. Thereby, we describe an alternate approach to finding a lower bound on the optimal error probability in the Chernoff-Stein lemma. Further, we show that the experiment selection strategy described in \cite{chernoff1959sequential, nitinawarat2013controlled} in a \emph{sequential} setting is asymptotically optimal for our fixed horizon problem. We also describe an alternate heuristic approach to strategy design that might improve the performance in the non-asymptotic regime. This approach is based on, what we call, the \emph{expected confidence rate} which naturally arises out of our analysis.

The rest of the paper is organized as follows. In Section \ref{priorwork}, we summarize key prior literature on hypothesis testing and discuss how our problem is related to various other formulations. In Section \ref{notation}, we describe our notation and in Section \ref{formulation}, we formulate our problem. We state the main results in Section \ref{main} and sketch the proof of our results in Section \ref{proofsec}. In Section \ref{discussion}, we discuss some heuristic approaches for strategy design. We conclude the paper in Section \ref{conc}.

\subsection{Prior Work}\label{priorwork}
Hypothesis testing is a long-standing problem and has been addressed in various settings. The classical formulations have been described in \cite{cover2012elements},\cite{chernoff1959sequential},\cite{wald1973sequential}. More recently, active hypothesis testing has been addressed in \cite{nitinawarat2013controlled},\cite{naghshvar2013active}.
The key difference between our formulation and the fixed horizon formulations in \cite{cover2012elements}, \cite{nitinawarat2013controlled} is that unlike our agent, the agents in these works are compelled to decide on a hypothesis after performing all the experiments. Our analysis shows that this modification significantly alters the optimal error exponents and the strategy design for inference and experiment selection. Another common formulation is the \emph{sequential} setting in which the agent can perform experiments until sufficiently strong evidence is gathered \cite{chernoff1959sequential}, \cite{nitinawarat2013controlled}, \cite{naghshvar2013active}. The objective in the sequential setting is to minimize a combination of Bayesian error probability and expected stopping time. Interestingly, the analysis and results in our fixed horizon problem have a strong overlap with those in the sequential setting. As mentioned earlier, a sub-problem in our analysis is a generalization of the Chernoff-Stein lemma \cite{cover2012elements} and our original problem can be seen as a symmetric version of this lemma. To the best of our knowledge, our formulation has not been considered before. The analysis involved in obtaining the upper bound for our problem borrows from prior works \cite{chernoff1959sequential},\cite{nitinawarat2013controlled}. However, our approach for obtaining lower bounds is different from the approach used in all the aforementioned works.

\subsection{Notation}\label{notation}
Random variables are denoted by upper case letters, their realization by the corresponding lower case letter. We use calligraphic fonts to denote sets (e.g. $\mathcal{U}$) and $\Delta \mathcal{U}$ is the probability simplex over a finite set $\mathcal{U}$. In general, subscripts denote time index unless stated otherwise. For time indices $n_1\leq n_2$, $\rv{Y}_{n_1:n_2}$ is the short hand notation for the variables $(\rv{Y}_{n_1},\rv{Y}_{n_1+1},...,\rv{Y}_{n_2})$.
For a strategy $g$, we use $\Py^g[\cdot]$ and $\E^g[\cdot]$ to indicate that the probability and expectation depend on the choice of $g$. For an hypothesis $i$, $\E_i^g[\cdot]$ denotes the expectation conditioned on hypothesis $i$. The Kullback-Leibler divergence between distributions $p$ and $q$ over a finite space $\mathcal{Y}$ is given by
\begin{equation}
D(p || q) = \sum_{y \in \mathcal{Y}}p(y)\log\frac{p(y)}{q(y)}.
\end{equation}

\section{Problem Formulation}\label{formulation}
Let $\mathcal{H} = \{1,2,\ldots,M\}$ be a finite set of hypotheses and let the random variable $\rv{H}$ denote the true hypothesis. The prior probability on $\rv{H}$ is $\vct{\rho}_1$. At each time $n=1,2,\ldots$, an agent can perform an experiment $\rv{U }_n \in \mathcal{U}$ and obtain an observation $\rv{Y}_n \in \mathcal{Y}$. We assume that the sets $\mathcal{U}$ and $\mathcal{Y}$ are finite.  The observation $\rv{Y}_n$ at time $n$ is given by
\begin{equation}
\rv{Y}_n = \xi(\rv{H}, \rv{U}_n,\rv{W}_n).
\end{equation}
where $\{\rv{W}_n:  n = 1,2,\dots\}$ is a collection of mutually independent and identically distributed primitive random variables. The probability of observing  $y$ after performing an experiment $u$ under hypothesis $h$ is denoted by $p_h^u(y)$, that is,
\[ p_h^u(y) := \Py(Y_n=y \mid H=h,U_n=u).  \]
 The time horizon, that is the total number of experiments performed, is fixed \emph{a priori} to $N < \infty$.

At time $n=1,2,\ldots$, the information available to the agent, denoted by $\rv{I}_n$, is the collection of all experiments performed and the corresponding observations up to time $n-1$, \emph{i.e}.
\begin{equation}
\rv{I}_n = \{\rv{U}_{1:n-1},\rv{Y}_{1:n-1}\}.
\end{equation}
At time $n$, the agent selects a distribution over the set of actions $\mathcal{U}$ according to an \emph{experiment selection rule} $g_n$ and the action $\rv{U}_n$ is randomly drawn from this distribution,
%The experiment selected at time $n$ can be chosen as a function of $\rv{I}_n$. Let the \emph{strategy} used for selecting the experiment be $g_n$,
that is
\begin{equation}
\rv{U}_n  \sim g_n(\rv{I}_n).
\end{equation}
The sequence  $\{g_n, n=1,\ldots,N\}$ is denoted by $g$ and  referred to as the \emph{experiment selection strategy}. Let the collection of all such strategies be $\mathcal{G}$. 

After performing $N$ experiments, the agent can declare one of the hypotheses  to be true or it can declare that its experiments were inconclusive. We refer to this final declaration as the  agent's \emph{inference decision} and denote it by  $\hat{\rv{H}}_N$.  The inference decision can take values in $\mathcal{H} \cup \{\varnothing\}$, where $\varnothing$ denotes the inconclusive declaration. $\hat{\rv{H}}_N$ is chosen according to an \emph{inference strategy} $f$,  \emph{i.e}.
\begin{equation}
    \hat{\rv{H}}_{N} = f(\rv{I}_{N+1}).
\end{equation}
Let the set of all inference strategies be $\mathcal{F}$. 

For an experiment selection strategy $g$ and an inference strategy $f$, we define the following error probabilities.
\begin{definition}
Let $\psi_N(i)$ be the probability that the agent does not infer $i$ when the true hypothesis is indeed $i$, \emph{i.e}.
\begin{align}
    \psi_N(i) &:= \Py^{f,g}[\hat{\rv{H}}_{N} \neq i \mid \rv{H} = i].\\
    %&=\sum_{j\neq i}\Py^{f,g}[\hat{\rv{H}}_{N} = j \mid \rv{H} = i] + \Py^{f,g}[\hat{\rv{H}}_{N} = \varnothing \mid \rv{H} = i]
    \shortintertext{We refer to $\psi_N(i)$ as type-$i$ error probability. Let $\phi_N(i)$ be the probability that the agent infers $i$ but the true hypothesis is not $i$, \emph{i.e}.}
    \phi_N(i) &:= \Py^{f,g}[\hat{\rv{H}}_{N} = i \mid \rv{H} \neq i].
\end{align}
\end{definition}
\begin{remark}
Note that when there are only two hypotheses and the agent is forced to decide on one of the two hypotheses, $\psi_N(1)$ and $\psi_N(2)$ are type I and type II errors, respectively. In this case, $\psi_N(1) = \phi_N(2)$ and $\psi_N(2) = \phi_N(1)$.
\end{remark}

In this paper, we will be interested in the event that the agent declares an incorrect hypothesis to be true.  That is, we will consider the event $\cup_{i \in \mathcal{H}} \{\hat{\rv{H}}_{N} = i, \rv{H} \neq i \}$. We refer to this event as the \emph{misclassification event}. Let $\gamma_N$ be the probability of this event. Using the definitions above, this probability can be written as
\begin{align}
    \gamma_N &=  \sum_{i \in \mathcal{H}}\Py^{f,g}[\hat{\rv{H}}_{N} = i \mid \rv{H} \neq i]\Py[\rv{H} \neq i]\\
    &= \sum_{i \in \mathcal{H}}\phi_N(i)(1-\rho_1(i)).
\end{align}
We will consider the problem of designing the experiment selection and inference strategies to minimize $\gamma_N$ (the probability of declaring an incorrect hypothesis) while satisfying constraints on the type-$i$ error probabilities. That is, we are interested in the following optimization problem:
%The objective is to design an inference strategy $f$ and an experiment selection strategy $g$ which are solutions to the following optimization problem.
\begin{align}\tag{P1}\label{opt1}
    & & & \underset{f \in \mathcal{F},g \in \mathcal{G}}{\text{min}} & & \gamma_N & \\
    \nonumber& & & \text{subject to} & & \psi_N(i) \leq \epsilon_N,\; \forall i \in \mathcal{H} & 
\end{align}
where $0 < \epsilon_N < 1$. Let $\gamma^*_N$ denote the infimum value of this optimization problem. We define $\gamma^*_N := \infty$ if the optimization problem is infeasible.

The above  formulation is intended for scenarios where the penalty for declaring an incorrect hypothesis to be true is much higher than the penalty for making no decision about the hypothesis. In such cases, it is reasonable for the agent to abstain from drawing conclusions when the evidence is not strong enough. The constraints on type-$i$ error probabilities ensure that  the agent does not abstain from drawing conclusions too often. The optimization problem seeks to minimize the probability of declaring an incorrect hypothesis while satisfying the type-$i$ error probability constraints.

\section{Main Results}\label{main}
In this section, we will describe asymptotically tight lower and upper bounds on the optimal error probability $\gamma_N^*$ in Problem (\ref{opt1}). We will first define some useful quantities and then state the assumptions we make to prove our results.

The \emph{posterior belief}  $\vct{\rho}_n$  on the hypothesis $\rv{H}$   based on information $\rv{I}_n$ is given by
\begin{equation}\label{postbelief}
\rho_n(i) = \Py[\rv{H} = i \mid \rv{U}_{1:n-1},\rv{Y}_{1:n-1}] =\Py[\rv{H} = i \mid \rv{I}_n].
\end{equation}
{Note that given a realization of the experiments and observations until time $n$,  the posterior belief does not depend on the experiment selection strategy $g$.}

\begin{definition}[\emph{Bayesian Log-Likelihood Ratio \& Expected Confidence Rate}]
%Let $\vct{\rho}\in \Delta\mathcal{H}$ be a distribution over $\mathcal{H}$. 
The \emph{Bayesian log-likelihood ratio} $\mathcal{C}_i(\vct{\rho})$ associated with an hypothesis $i \in \mathcal{H}$ is defined as
\begin{equation}
\mathcal{C}_i(\vct{\rho}) := \log\frac{\rho(i)}{1-\rho(i)}.\\
\end{equation}
The Bayesian log-likelihood ratio (BLLR) is the logarithm of the ratio of the probability that hypothesis $i$ is true versus the probability that hypothesis $i$ is not true. The BLLR can be interpreted as a \emph{confidence level} on hypothesis $i$. For a hypothesis $i$ and a strategy $g \in \mathcal{G}$, we define the expected confidence rate $J_N^g(i)$ as
\begin{align}\label{eq:jng}
&J_N^g(i) := \frac{1}{N} \E_i^{g} \left[\mathcal{C}_{i}(\vct{\rho}_{N+1})- \mathcal{C}_{i}(\vct{\rho}_1) \right].
%=&\lim_{N \to \infty} \inf \frac{1}{N} \;\E^{g'} \left[\mathcal{C}_{h}(\vct{\rho}(N+1))- \mathcal{C}_{h}(\vct{\rho}(1)) \mid \rv{H} = h\right].
\end{align}
\end{definition}
%\vspace{0.05in}

\begin{assumption}[\emph{Full support}]\label{boundedassump}
There exists a constant $B>0$ such that $|\lambda_j^i(u,y)| < B$ for every experiment $u \in \mathcal{U}$, observation $y \in \mathcal{Y}$ and pair of hypotheses $i,j \in \mathcal{H}$, where
$$
\lambda_j^i(u,y) := \log\frac{p_i^u(y)}{p_j^u(y)}.
$$
\end{assumption}

\begin{assumption}\label{steadinf}
For each experiment $u \in \mathcal{U}$ and any pair of hypotheses $i,j \in \mathcal{H}$ such that $i \neq j$, we have
\begin{align}
D(p_i^u || p_j^u) > 0.
\end{align}
\end{assumption}
\begin{remark}
We make Assumption \ref{steadinf} for ease of exposition. Techniques for relaxing this Assumption have been discussed in \cite{nitinawarat2013controlled} and \cite{naghshvar2013active}.
\end{remark}
\noindent  For each hypothesis $i \in \mathcal{H}$, define
\begin{align}\label{eq:dstarti}
D^*(i) &:= \max_{\vct{\alpha} \in \Delta\mathcal{U}} \min_{j\neq i} \sum_{u}\alpha(u) D(p_i^u || p_j^u) \\
&=  \min_{\vct{\beta} \in \Delta\tilde{\mathcal{H}}_i} \max_{u \in \mathcal{U}} \sum_{j \neq i}\beta(j) D(p_i^u || p_j^u),
\end{align}
where $\tilde{\mathcal{H}}_i = \mathcal{H} \setminus \{i\}$. The equality of the min-max and max-min values follows from the minimax theorem \cite{osborne1994course} because the sets $\mathcal{U}$ and $\mathcal{H}$ are finite and the Kullback-Leibler divergences are bounded by $B$.

\begin{assumption}\label{epsassum}
For every $N \geq 1$, we have that the bound on the type-$i$ error satisfies $0 <\epsilon_N \leq 1/2N$. Further,
\begin{align}
    \lim_{N \to \infty}\frac{-\log{\epsilon_N}}{N} = 0.
\end{align}
\end{assumption}

\begin{theorem}[Lower bound]\label{lbthm} 
There exists a positive constant $K_1$ that does not depend on $N$ such that for every $N \geq 1$ the following statements are true.
\begin{enumerate}[label=\alph*)]
\item For any experiment selection strategy $g$ and inference strategy $f$ that satisfy the constraints $\psi_N(i) \leq \epsilon_N$ for every $i \in \mathcal{H}$, we have the lower bound
\begin{align}
   \gamma_N &\geq \sum_{i \in \mathcal{H}}(1-\rho_1(i))\exp(-NJ^g_N(i) -  K_1),
\end{align}
where $J^g_N(i)$ is given by \eqref{eq:jng}. 
\item The optimal misclassification probability $\gamma_N^*$ in Problem (\ref{opt1}) satisfies
\begin{align}
   \gamma^*_N &\geq \sum_{i \in \mathcal{H}}(1-\rho_1(i))\exp(-ND^*(i) -  K_1),
\end{align}
where $D^*(i)$ is given by \eqref{eq:dstarti}.
\end{enumerate}
\end{theorem}

\begin{theorem}[Upper bound]\label{achieve}
For any $\delta > 0$, there exists an integer $N_\delta$ such that for every $N\geq N_\delta$, we have
\begin{align}
    \gamma^*_N \leq \sum_{i \in \mathcal{H}}(1-\rho_1(i))\exp(-N(D^*(i) -\delta)).
\end{align}
\end{theorem}

Using Theorems \ref{lbthm} and \ref{achieve}, we can therefore conclude that
\begin{align}
\lim_{N \to \infty}-\frac{1}{N}\log\gamma_N^* = \min_{i\in \mathcal{H}}D^*(i).
\end{align}

% \begin{theorem}
% We have
% $$R^*_N(i) > D^*(i) - B\sqrt{\frac{-\log{\epsilon_N}}{N}}.$$
% \end{theorem}
% % \begin{theorem}
% % We have
% % $$\frac{1}{N}[\mathcal{C}_i(\vct{\rho}_{N+1}) - \mathcal{C}_i(\vct{\rho}_{1})] > D^*(i) - B\sqrt{\frac{\log{\epsilon_N}}{N}}.$$
% % \end{theorem}

% \begin{theorem}
% Let $g$ be any experiment selection strategy and let $f$ be any inference strategy such that
% $\psi_N(i) \leq  \epsilon_N.$ Then
% $$ R^{f,g}_N(i)  \leq J_{N}^g(i) + \mathcal{O}(\epsilon_N + 1/N).$$
% \end{theorem}

% \begin{theorem}

% \end{theorem}

\section{Proof of Main Results}\label{proofsec}

\subsection{Supporting Lemmas}
In this section, we describe some important properties of the confidence level $\mathcal{C}_i(\vct{\rho})$ which will be used in the proof of our main results.

\begin{lemma}\label{logsumexplemma}
For any experiment selection strategy $g$, we have
\begin{align}
&\nonumber\mathcal{C}_i(\vct{\rho}_{N+1}) - \mathcal{C}_i(\vct{\rho}_1) \\
&=  - \log\left[\sum_{j\neq i}\exp\left(\log\tilde{\rho}_1(j) +  \sum_{n=1}^N\lambda_i^j(\rv{U}_n,\rv{Y}_n)\right)\right],
\end{align}
where $\tilde{\rho}_1(j) = \rho_1(j)/(1-\rho_1(i)).$
\end{lemma}
%\begin{proof}
%The proof follows from simple algebraic manipulations. We refer the reader to \cite{prooflink} for details.
%\end{proof}
 \begin{proof}
 We have
 \begin{align}
 &\log\frac{\rho_{N+1}(i)}{1-\rho_{N+1}(i)}-\log\frac{\rho_1(i)}{1-\rho_1(i)} \\
 = &\log \frac{\rho_1(i)\prod_{n=1}^Np_i^{\rv{U}_n}(\rv{Y}_n)}{\sum_{j\neq i}\rho_1(j)\prod_{n=1}^Np_j^{\rv{U}_n}(\rv{Y}_n)}-\log\frac{\rho_1(i)}{1-\rho_1(i)} \\
 = &\log \frac{\prod_{n=1}^Np_i^{\rv{U}_n}(\rv{Y}_n)}{\sum_{j\neq i}\tilde{\rho}_1(j)\prod_{n=1}^Np_j^{\rv{U}_n}(\rv{Y}_n)} \\
 = & -\log\sum_{j\neq i}\tilde{\rho}_1(j)\frac{\prod_{n=1}^Np_j^{\rv{U}_n}(\rv{Y}_n)}{\prod_{n=1}^Np_i^{\rv{U}_n}(\rv{Y}_n)}\\
 \nonumber = & - \log\sum_{j\neq i}\exp(\log\tilde{\rho}_1(j) +  \sum_{n=1}^N\lambda_1^j(\rv{U}_n,\rv{Y}_n)).
 \end{align}
 \end{proof}

\begin{corollary}[\emph{Bounded increments}]\label{boundedcoro}
For any experiment selection strategy $g \in \mathcal{G}$, we have
\begin{align}
    |\mathcal{C}_i(\vct{\rho}_{N+1}) - \mathcal{C}_i(\vct{\rho}_1)| < NB,
\end{align}
with probability 1.
\end{corollary}
%\begin{proof}
%This follows from Lemma \ref{logsumexplemma} and Assumption \ref{boundedassump}. See \cite{prooflink} for details.
%\end{proof}

 \begin{proof}
 Using Lemma \ref{logsumexplemma}, we have
 \begin{align}
 \mathcal{C}_i(\vct{\rho}_{N+1}) - \mathcal{C}_i(\vct{\rho}_{1}) &=  - \log\sum_{j\neq i}\exp(\log\tilde{\rho}_1(j) +  \sum_{n=1}^N\lambda_i^j(\rv{U}_n,\rv{Y}_n))\\
 \shortintertext{Using Assumption \ref{boundedassump}, it follows that}
  &\leq - \log\sum_{j\neq i}\exp(\log\tilde{\rho}_1(j) -NB)\\
   &= - \log\sum_{j\neq i}\exp(\log\tilde{\rho}_1(j)) + NB\\
   &= - \log\sum_{j\neq i}\tilde{\rho}_j(1) + NB\\
   &= NB.
 \end{align}
 Similarly,
 \begin{align}
 \mathcal{C}_i(\vct{\rho}_{N+1}) - \mathcal{C}_i(\vct{\rho}_{1}) &=  - \log\sum_{j\neq i}\exp(\log\tilde{\rho}_1(j) +  \sum_{n=1}^N\lambda_i^j(\rv{U}_n,\rv{Y}_n))\\
 \shortintertext{Using Assumption \ref{boundedassump}, it follows that}
  &\geq - \log\sum_{j\neq i}\exp(\log\tilde{\rho}_1(j) +NB)\\
   &= - \log\sum_{j\neq i}\exp(\log\tilde{\rho}_1(j)) - NB\\
   &= - \log\sum_{j\neq i}\tilde{\rho}_j(1) - NB\\
   &= -NB.
 \end{align}
 
The same arguments can be used to show that for any experiment selection strategy $g \in \mathcal{G}$ and $1 \leq n \leq N$, we have
\begin{align}
    |\mathcal{C}_i(\vct{\rho}_{n+1}) - \mathcal{C}_i(\vct{\rho}_n)| < B,
\end{align}
with probability 1.
 \end{proof}

\begin{corollary}\label{cortriv}
If for every $j \neq i$, $ \sum_{n=1}^N\lambda_j^i(\rv{U}_n,\rv{Y}_n)) \geq \theta$ for some $\theta \in \R$, then
\begin{align}
    \mathcal{C}_i(\vct{\rho}_{N+1}) - \mathcal{C}_i(\vct{\rho}_{1}) \geq \theta.
\end{align}
\end{corollary}
\begin{proof}
We have
\begin{align}
\mathcal{C}_i(\vct{\rho}_{N+1}) - \mathcal{C}_i(\vct{\rho}_{1}) &=  - \log\sum_{j\neq i}\exp(\log\tilde{\rho}_1(j) + \sum_{n=1}^N\lambda_i^j(\rv{U}_n,\rv{Y}_n))\\
 &\geq - \log\sum_{j\neq i}\exp(\log\tilde{\rho}_1(j) - \theta)\\
 &= - [\log\sum_{j\neq i}\exp(\log\tilde{\rho}_1(j))] + \theta\\
 &= - [\log\sum_{j\neq i}\tilde{\rho}_j(1)] + \theta\\
 &= \theta.
\end{align}
\end{proof}

\begin{lemma}\label{condexp}
For any experiment selection strategy $g$,
\begin{align}
  &\E^g_i  \left[\sum_{n=1}^N\lambda_j^i(\rv{U}_n,\rv{Y}_n)\right]= \E^g_i \left[\sum_{n=1}^N D(p_i^{\rv{U}_n}||p_j^{\rv{U}_n})\right].
\end{align}
\end{lemma}

 \begin{proof}
 \begin{align}
     \E^g_i  \sum_{n=1}^N\lambda_j^i(\rv{U}_n,\rv{Y}_n) &= \E^g_i\sum_{n=1}^N\E_i[\lambda_j^i(\rv{U}_n,\rv{Y}_n)\mid \rv{U}_n]\\
     &= \E^g_i\sum_{n=1}^N D(p_i^{\rv{U}_n}||p_j^{\rv{U}_n}).
 \end{align}
 The last inequality follows from the fact that the observations $\rv{Y}_n$ are independent conditioned on the experiment $\rv{U}_n$.
 \end{proof}

\begin{definition}
We define the following distributions:
\begin{align}
\vct{\alpha}^{i*} &:= \argmax_{\vct{\alpha} \in \Delta\mathcal{U}} \min_{j\neq i} \sum_{u}\alpha(u) D(p_i^u || p_j^u),\\
\vct{\beta}^{i*} &:= \argmin_{\vct{\beta} \in \Delta\tilde{\mathcal{H}}_i} \max_{u \in \mathcal{U}} \sum_{j \neq 1}\beta(j) D(p_i^u || p_j^u),
\end{align}
where $\tilde{\mathcal{H}}_i = \mathcal{H} \setminus \{i\}$.
\end{definition}

\begin{lemma}\label{kldbound}
For any experiment selection strategy $g$, we have
\begin{align}
    J_N^g(i) \leq D^*(i) - \frac{\sum_{j \in \tilde{\mathcal{H}}_i}\beta^{i*}(j)\log\tilde{\rho}_1(j)}{N},
\end{align}
where $\tilde{\rho}_1(j) = \rho_1(j)/(1-\rho_1(i)).$
\end{lemma}
%\begin{proof}
%See \cite{prooflink} for a proof.
%\end{proof}
 \begin{proof}
 For every $j \neq i$, we have the following since $\log x$ is an increasing function 
 \begin{align}
     &\nonumber\mathcal{C}_i(\vct{\rho}_{N+1}) - \mathcal{C}_i(\vct{\rho}_1) \\
      & = - \log\sum_{j\neq i}\exp(\log\tilde{\rho}_1(j) +  \sum_{n=1}^N\lambda_i^j(\rv{U}_n,\rv{Y}_n))\\
      &\leq - \log\exp(\log\tilde{\rho}_1(j) +  \sum_{n=1}^N\lambda_i^j(\rv{U}_n,\rv{Y}_n))\\
      &= - \log\tilde{\rho}_1(j) +  \sum_{n=1}^N\lambda_j^i(\rv{U}_n,\rv{Y}_n).
 \end{align}
 Therefore,
 \begin{align}
      &\nonumber\mathcal{C}_i(\vct{\rho}_{N+1}) - \mathcal{C}_i(\vct{\rho}_1) \\
      &\leq \sum_{j\neq i}\beta^{i*}(j)[- \log\tilde{\rho}_1(j) +  \sum_{n=1}^N\lambda_j^i(\rv{U}_n,\rv{Y}_n)].
 \end{align}
 Further,
 \begin{align}
     J_N^g(i)  &= \frac{1}{N} \E^{g} \left[\mathcal{C}_{i}(\vct{\rho}_{N+1})- \mathcal{C}_{i}(\vct{\rho}_1) \mid \rv{H} = i\right]\\
     &\leq \frac{1}{N}\E^g\left[\sum_{j\neq i}\beta^{i*}(j)[- \log\tilde{\rho}_1(j) +  \sum_{n=1}^N\lambda_j^i(\rv{U}_n,\rv{Y}_n)]\mid \rv{H} = i\right]\\
     &= - \frac{\sum_{j \in \tilde{\mathcal{H}}_i}\beta^{i*}(j)\log\tilde{\rho}_1(j)}{N} + \frac{1}{N}\E^g\left[\sum_{j\neq i}\beta^{i*}(j)[\sum_{n=1}^N\lambda_j^i(\rv{U}_n,\rv{Y}_n)]\mid \rv{H} = i\right]\\
     \shortintertext{Using Lemma \ref{condexp}, we have}
     &= - \frac{\sum_{j \in \tilde{\mathcal{H}}_i}\beta^{i*}(j)\log\tilde{\rho}_1(j)}{N} + \frac{1}{N}\E^g\left[\sum_{j\neq i}\beta^{i*}(j)[\sum_{n=1}^ND(p_i^{\rv{U}_n}||p_j^{\rv{U}_n})]\mid \rv{H} = i\right]\\
     &= - \frac{\sum_{j \in \tilde{\mathcal{H}}_i}\beta^{i*}(j)\log\tilde{\rho}_1(j)}{N} + \frac{1}{N}\E^g\left[\sum_{n=1}^N\sum_{j\neq i}\beta^{i*}(j)D(p_i^{\rv{U}_n}||p_j^{\rv{U}_n})\mid \rv{H} = i\right]\\
     \shortintertext{Since $\vct{\beta}^{i*}$ is the minimax distribution, we have}
     &\leq - \frac{\sum_{j \in \tilde{\mathcal{H}}_i}\beta^{i*}(j)\log\tilde{\rho}_1(j)}{N} + \frac{1}{N}\E^g\left[\sum_{n=1}^N D^*(i)\mid \rv{H} = i\right]\\
     &= - \frac{\sum_{j \in \tilde{\mathcal{H}}_i}\beta^{i*}(j)\log\tilde{\rho}_1(j)}{N} + D^*(i).
 \end{align}
 \end{proof}

\begin{lemma}\label{thresh}
Let $f$ be an inference strategy in which hypothesis $i$ is decided if and only if $\mathcal{C}_i(\vct{\rho}_{N+1}) - \mathcal{C}_i(\vct{\rho}_1) \geq \theta$. Then
\begin{equation}
   \phi_N(i) \leq e^{-\theta}.
\end{equation}
\end{lemma}
%\begin{proof}
%See \cite{prooflink} for details.
%\end{proof}

\begin{proof}
In this proof, substitute $n$ with $N$. Let $Z_n$ be the region in which the inference policy $f$ selects hypothesis $i$, that is
$$Z_{n+1} := \{\iota_{n+1}: f_{n+1}(\iota_{n+1}) = i \text{ and } \Py^{f,g}[\rv{I}_{n+1} = \iota_{n+1}] \neq 0\}. $$
We have
\begin{align}
    &\Py^{f,g}[\hat{\rv{H}}_{n+1} = i \text{ and } \rv{H} \neq i ] \\
    &= \Py^{g}[\rv{I}_{n+1} \in Z_{n+1} \text{ and } \rv{H} \neq i ]\\
    &= \sum_{\iota_{n+1} \in Z_{n+1}}\Py^{g}[\rv{I}_{n+1} = \iota_{n+1} \text{ and } \rv{H} \neq i ]\\
    &= \sum_{\iota_{n+1} \in Z_{n+1}}\Py^{g}[\rv{I}_{n+1} = \iota_{n+1} \text{ and } \rv{H} = i ]\exp\left[-\log\frac{\Py^{g}[\rv{I}_{n+1} = \iota_{n+1} \text{ and } \rv{H} = i ]}{\Py^{g}[\rv{I}_{n+1} = \iota_{n+1} \text{ and } \rv{H} \neq i ]}\right]\\
    &= \sum_{\iota_{n+1} \in Z_{n+1}}\Py^{g}[\rv{I}_{n+1} = \iota_{n+1} \text{ and } \rv{H} = i ]\exp\left[-\mathcal{C}_i(\iota_{n+1})\right]\\
    &\leq \sum_{\iota_{n+1} \in Z_{n+1}}\Py^{g}[\rv{I}_{n+1} = \iota_{n+1} \text{ and } \rv{H} = i ]\exp\left[-(\theta + \mathcal{C}_i(\vct{\rho}_1))\right]\\
    &\leq \rho_1(i)e^{-(\theta + \mathcal{C}_i(\vct{\rho}_1))}.
\end{align}
Therefore,
\begin{equation}
    \Py^{f,g}[\hat{\rv{H}}_{n+1} = i \mid \rv{H} \neq i ] \leq e^{-\theta}.
\end{equation}
\end{proof}

\subsection{Sub-problem vis-\`a-vis Chernoff-Stein}
We formulate a sub-problem in this section that will be useful for analyzing Problem (\ref{opt1}).
For hypothesis $i\in \mathcal{H}$, consider the following optimization problem: 
\begin{align}\tag{P2}\label{opti}
    & & & \underset{f \in \mathcal{F},g \in \mathcal{G}}{\text{min}} & & \phi_N(i) & \\
    \nonumber & & & \text{subject to} & & \psi_N(i) \leq \epsilon_N. & 
\end{align}
Let the infimum value of this optimization problem be $\phi^*_N(i)$. 
Note that this problem is always feasible because the agent can trivially satisfy the type-$i$ error constraint by always declaring hypothesis $i$.
\begin{remark}
When there is only one experiment, two hypotheses and the inconclusive decision $\varnothing$ is not allowed, this formulation is identical to that of the Chernoff-Stein lemma \cite{cover2012elements}.
\end{remark}

We follow the proof methodology of the Chernoff-Stein lemma in \cite{cover2012elements}, but with some important modifications. For each experiment selection strategy $g$ and inference strategy $f$ that satisfy the type-$i$ error constraint, we first establish a lower bound on the error probability $\phi_N(i)$ based on the expected confidence rate $J_N^g$. {In \cite{cover2012elements}, the lower bound is obtained using a typicality argument. However, such typicality properties may not hold for every experiment selection strategy $g$.} Thus, we use a different approach to obtain a similar lower bound. We then use Lemma \ref{kldbound} to obtain a lower bound on $\phi_N(i)$ that does not depend on the strategies $g$ and $f$. Further, we construct strategies that asymptotically achieve this strategy-independent lower bound. The construction of these strategies and the analysis thereof builds on the achievability proofs in \cite{cover2012elements} and \cite{chernoff1959sequential}.
\begin{lemma}\label{cslb}
Let $g$ be any experiment selection strategy and let $f$ be any inference strategy such that
$\psi_N(i) \leq  \epsilon_N.$ Then
\begin{align}\label{eq:cslb}
    -\frac{1}{N}\log\phi_N(i)  &\leq J_{N}^g(i) + \frac{2B\epsilon_N }{1-\epsilon_N}-\frac{1}{N}\log(1-\epsilon_N).
\end{align}
\end{lemma}
\begin{proof}
In this proof, substitute $n$ with $N$ and $\epsilon$ with $\epsilon_N$. Also, if the belief $\vct{\rho}_{n+1}$ is formed using information $\iota_{n+1}$, we denote $\mathcal{C}_i(\vct{\rho}_{n+1})$ with $\mathcal{C}_i(\iota_{n+1})$ to emphasize dependence on $\iota_{n+1}$. Let $Z_n$ be the region in which the inference policy $f$ selects hypothesis $i$, that is
$$Z_{n+1} := \{\iota_{n+1}: f_{n+1}(\iota_{n+1}) = i \text{ and } \Py^{f,g}[\rv{I}_{n+1} = \iota_{n+1}] \neq 0\}. $$
We have
\begin{align}
    &\Py^{f,g}[\hat{\rv{H}}_{n+1} = i \text{ and } \rv{H} \neq i ] \\
    &= \Py^{g}[\rv{I}_{n+1} \in Z_{n+1} \text{ and } \rv{H} \neq i ]\\
    &= \sum_{\iota_{n+1} \in Z_{n+1}}\Py^{g}[\rv{I}_{n+1} = \iota_{n+1} \text{ and } \rv{H} \neq i ]\\
    &= \sum_{\iota_{n+1} \in Z_{n+1}}\Py^{g}[\rv{I}_{n+1} = \iota_{n+1} \text{ and } \rv{H} = i ]\exp\left[-\log\frac{\Py^{g}[\rv{I}_{n+1} = \iota_{n+1} \text{ and } \rv{H} = i ]}{\Py^{g}[\rv{I}_{n+1} = \iota_{n+1} \text{ and } \rv{H} \neq i ]}\right]\\
    &= \sum_{\iota_{n+1} \in Z_{n+1}}\Py^{g}[\rv{I}_{n+1} = \iota_{n+1} \text{ and } \rv{H} = i ]\exp\left[-\mathcal{C}_i(\iota_{n+1})\right]\\
    &= \Py^{g}[\rv{I}_{n+1} \in Z_{n+1} \text{ and } \rv{H} = i ]\sum_{\iota_{n+1} \in Z_{n+1}}\frac{\Py^{g}[\rv{I}_{n+1} = \iota_{n+1} \text{ and } \rv{H} = i ]}{\Py^{g}[\rv{I}_{n+1} \in Z_{n+1} \text{ and } \rv{H} = i ]}\exp\left[-\mathcal{C}_i(\iota_{n+1})\right]\\
    &= \Py^{g}[\rv{I}_{n+1} \in Z_{n+1} \text{ and } \rv{H} = i ]\sum_{\iota_{n+1} \in Z_{n+1}}{\Py^{g}[\rv{I}_{n+1} = \iota_{n+1} \mid \rv{I}_{n+1} \in Z_{n+1}\text{ and } \rv{H} = i ]}\exp\left[-\mathcal{C}_i(\iota_{n+1})\right]\\
    &= \Py^{g}[\rv{I}_{n+1} \in Z_{n+1} \text{ and } \rv{H} = i ]{\E^{g}[\exp\left[-\mathcal{C}_i(\rv{I}_{n+1})\right] \mid \rv{I}_{n+1} \in Z_{n+1}\text{ and } \rv{H} = i ]}.
\end{align}
The function $-\log x$ is convex in $x$ and thus, using Jensen's inequality, we have
\begin{align}
    &-\frac{1}{n}\log\Py^{f,g}[\hat{\rv{H}}_{n+1} = i \text{ and } \rv{H} \neq i ] \\
    &\leq -\frac{1}{n}\log\Py^{g}[\rv{I}_{n+1} \in Z_{n+1} \text{ and } \rv{H} = i ]+\frac{1}{n}{\E^{g}[\mathcal{C}_i(\rv{I}_{n+1}) \mid \rv{I}_{n+1} \in Z_{n+1}\text{ and } \rv{H} = i ]}\\
    &= -\frac{1}{n}\log\Py^{f,g}[\hat{\rv{H}}_{n+1} = i \text{ and } \rv{H} = i ]+\frac{1}{n}{\E^{g}[\mathcal{C}_i(\rv{I}_{n+1}) \mid \rv{I}_{n+1} \in Z_{n+1}\text{ and } \rv{H} = i ]}\\
    &= -\frac{1}{n}\log\Py^{f,g}[\hat{\rv{H}}_{n+1} = i \mid \rv{H} = i ] - \frac{1}{n}\log \rho_1(i) +  \frac{1}{n}{\E^{g}[\mathcal{C}_i(\rv{I}_{n+1}) \mid \rv{I}_{n+1} \in Z_{n+1}\text{ and } \rv{H} = i ]}\\
    \label{start1}&\leq -\frac{1}{n}\log(1-\epsilon) - \frac{1}{n}\log \rho_1(i) +  \frac{1}{n}{\E^{g}[\mathcal{C}_i(\rv{I}_{n+1}) \mid \rv{I}_{n+1} \in Z_{n+1}\text{ and } \rv{H} = i ]}.
\end{align}
Further, we have
\begin{align}
    nJ_{n}^g(i) + \mathcal{C}_i(\vct{\rho}_1) = {\E^{g}[\mathcal{C}_i(\rv{I}_{n+1}) \mid \rv{H} = i ]} &= \Py^{f,g}[\hat{\rv{H}}_{n+1} = i \mid \rv{H} = i ]{\E^{g}[\mathcal{C}_i(\rv{I}_{n+1}) \mid \rv{I}_{n+1} \in Z_{n+1}\text{ and } \rv{H} = i ]} \\
    &+ \Py^{f,g}[\hat{\rv{H}}_{n+1} \neq i \mid \rv{H} = i ]{\E^{g}[\mathcal{C}_i(\rv{I}_{n+1}) \mid \rv{I}_{n+1} \notin Z_{n+1}\text{ and } \rv{H} = i ]}.
\end{align}
Therefore,
\begin{align}
    &\frac{1}{n}{\E^{g}[\mathcal{C}_i(\rv{I}_{n+1}) \mid \rv{I}_{n+1} \in Z_{n+1}\text{ and } \rv{H} = i ]} \\
    &= \frac{1}{n}\times\frac{nJ_n^g(i) +\mathcal{C}_i(\vct{\rho}_1)-\Py^{f,g}[\hat{\rv{H}}_{n+1} \neq i \mid \rv{H} = i ]{\E^{g}[\mathcal{C}_i(\rv{I}_{n+1}) \mid \rv{I}_{n+1} \notin Z_{n+1}\text{ and } \rv{H} = i ]}}{\Py^{f,g}[\hat{\rv{H}}_{n+1} = i \mid \rv{H} = i ]}\\
    &= \frac{J_n^g(i)+\frac{1}{n}\mathcal{C}_i(\vct{\rho}_1)-\frac{1}{n}\Py^{f,g}[\hat{\rv{H}}_{n+1} \neq i \mid \rv{H} = i ]{\E^{g}[\mathcal{C}_i(\rv{I}_{n+1}) \mid \rv{I}_{n+1} \notin Z_{n+1}\text{ and } \rv{H} = i ]}}{\Py^{f,g}[\hat{\rv{H}}_{n+1} = i \mid \rv{H} = i ]}\\
    \shortintertext{Since $\Py^{f,g}[\hat{\rv{H}}_{n+1} = i \mid \rv{H} = i ] \geq 1-\epsilon$ and from Corollary \ref{boundedcoro}, we have }
    &\leq \frac{J_n^g(i)+\frac{1}{n}\mathcal{C}_i(\vct{\rho}_1)+\epsilon (B - \frac{1}{n}\mathcal{C}_i(\vct{\rho}_{1}))}{1-\epsilon}\\
    &= \frac{J_n^g(i)+\epsilon B }{1-\epsilon}+\frac{1}{n}\mathcal{C}_i(\vct{\rho}_1)\\
    &= J_n^g(i) + \frac{\epsilon J_n^g(i)+\epsilon B }{1-\epsilon}+\frac{1}{n}\mathcal{C}_i(\vct{\rho}_1)\\
    &\leq J_n^g(i) + \frac{2\epsilon B }{1-\epsilon}+\frac{1}{n}\mathcal{C}_i(\vct{\rho}_1).
\end{align}
The last inequality follows once again from Corollary \ref{boundedcoro}. Hence using inequality (\ref{start1}), we have
\begin{align}
    -\frac{1}{n}\log\Py^{f,g}[\hat{\rv{H}}_{n+1} = i \text{ and } \rv{H} \neq i ] &\leq J_n^g(i) + \frac{2\epsilon B }{1-\epsilon}+\frac{1}{n}\mathcal{C}_i(\vct{\rho}_1)-\frac{1}{n}\log(1-\epsilon) - \frac{1}{n}\log \rho_1(i)\\
    &= J_n^g(i) + \frac{2\epsilon B }{1-\epsilon}-\frac{1}{n}\log(1-\epsilon) - \frac{1}{n}\log (1-\rho_1(i)).
\end{align}
Therefore,
\begin{align}
    -\frac{1}{n}\log\Py^{f,g}[\hat{\rv{H}}_{n+1} = i \mid \rv{H} \neq i ] &\leq J_n^g(i) + \frac{2\epsilon B }{1-\epsilon}-\frac{1}{n}\log(1-\epsilon).
\end{align}

\end{proof}

%\begin{proof}
%We will provide a brief outline of this proof. For details, see \cite{prooflink}. For strategies $f,g$ define the following set
%$$Z_{N+1} := \{\iota_{N+1}: f(\iota_{N+1}) = i, \Py^{f,g}[\rv{I}_{N+1} = \iota_{N+1}] \neq 0\}, $$
%and define the event
%$A_N := \{ \rv{I}_{N+1} \in Z_{N+1}, \rv{H} = i\}. $
%For any instance of information $\iota_{N+1}$, we have
%\begin{align}
%\mathcal{C}_i(\vct{\rho}_{N+1}) = \log\frac{\Py^{g}[\rv{I}_{N+1} = \iota_{N+1} , \rv{H} = i ]}{\Py^{g}[\rv{I}_{N+1} = \iota_{N+1} , \rv{H} \neq i ]}.
%\end{align} 
%With some algebra, we can conclude that
%\begin{align*}
%(1-\rho_1&(i))\phi_N(i)\\
%%&= \Py^{f,g}[\hat{\rv{H}}_{N+1} = i, \rv{H} \neq i ]\\
%%&= \Py^{g}[\rv{I}_{N+1} \in Z_{N+1}, \rv{H} \neq i ]\\
%%&= \Py^{f,g}[A_N]\,{\E^{f,g}[\exp\left[-\mathcal{C}_i(\vct{\rho}_{N+1})\right] \mid A_N]}\\
%&= \rho_1(i)(1-\psi_N(i)){\E^{f,g}[\exp\left[-\mathcal{C}_i(\vct{\rho}_{N+1})\right] \mid A_N]}.
%\end{align*}
%Using Jensen's inequality, we have
%\begin{align*}
%&-\log\phi_N(i)\\
%&\leq -\log(1-\psi_N(i))  + {\E^{f,g}[\mathcal{C}_i(\vct{\rho}_{N+1})- \mathcal{C}_i(\vct{\rho}_1) \mid A_N]}.
%\end{align*}
%This gives us a lower bound on $\phi_N(i)$ but notice that the expectation is conditioned on the event $A_N$. Using the fact that $\psi_N(i) \leq \epsilon_N$ and the bounded increments property (Corollary \ref{boundedcoro}) we can obtain the lower bound (\ref{eq:cslb}). See \cite{prooflink} for details.
%\end{proof}

\begin{lemma}\label{cslb2}
Let $g$ be any experiment selection strategy and let $f$ be any inference strategy such that
$\psi_N(i) \leq  \epsilon_N.$ Then there exist positive constants $K_1(i) \leq K'_1(i) $ that do not depend on $N$ such that
\begin{align}
    -\frac{1}{N}\log\phi_N(i)  &\leq J^g_N(i) + \frac{K_1(i)}{N}
    \leq D^*(i) + \frac{K'_1(i)}{N}.
\end{align}
\end{lemma}
\begin{proof}
This follows directly from Lemmas \ref{kldbound} and \ref{cslb}, and the fact that $\epsilon_N \leq 1/2N$.
\end{proof}

\begin{lemma}\label{csub2}
There exists an integer $N_i$ such that for every $N \geq N_i$
\begin{align}
    -\frac{1}{N}\log\phi^*_N(i)> D^*(i) - 2B\sqrt{\frac{1}{N}\log\frac{M}{\epsilon_N}}.
\end{align}
\end{lemma}
\begin{proof}
We prove this by constructing an experiment selection strategy and an inference strategy that achieve the rate and constraints. Let the agent select experiments randomly and independently from the distribution $\vct{\alpha}^{i*}$. Under strategy, we have for every $j \neq i$
\begin{align}
    \E_i[\lambda_j^i(\rv{U}_n,\rv{Y}_n)] \geq D^*(i).
\end{align}
Using Hoeffding's inequality, we have
\begin{align}
    &\Py_i[ \sum_{n=1}^N\lambda_j^i(\rv{U}_n,\rv{Y}_n) <  ND^*(i) - 2B\sqrt{N\log\frac{M}{\epsilon_N}}]\\
    &\leq \Py_i[ \sum_{n=1}^N\lambda_j^i(\rv{U}_n,\rv{Y}_n) <  N\E_i[\lambda_j^i(\rv{U}_n,\rv{Y}_n)] - 2B\sqrt{N\log\frac{M}{\epsilon_N}}]\tag{Since $\E_i[\lambda_j^i(\rv{U}_n,\rv{Y}_n)] \geq D^*(i)$}\\
    &\leq \exp(-\frac{2\times 4NB^2\log\frac{M}{\epsilon_N}}{4NB^2})\\
    &\leq \frac{\epsilon_N}{M}.\label{hoeff}
\end{align}
Let the inference policy be as follows. If $\mathcal{C}_i(\vct{\rho}_{N+1}) \geq ND^*(i) - 2B\sqrt{N\log\frac{M}{\epsilon_N}} + \mathcal{C}_i(\vct{\rho}_{1})$, decide hypothesis $i$. Otherwise, declare $\varnothing$. From Lemma \ref{thresh}, we have
\begin{align}
    \Py^{f,g}[\hat{\rv{H}}_{N+1} = i \text{ and } \rv{H} \neq i ] \leq \rho_1(i)e^{-(ND^*(i) - 2B\sqrt{N\log\frac{M}{\epsilon_N}} + \mathcal{C}_i(\vct{\rho}_{1}))}.
\end{align}
And thus,
\begin{align}
    -\frac{1}{N}\log\Py^{f,g}[\hat{\rv{H}}_{N+1} = i \mid \rv{H} \neq i ] \geq D^*(i) - 2B\sqrt{\frac{1}{N}\log\frac{M}{\epsilon_N}}.
\end{align}
Now we need to show that $\Py^{f,g}[\hat{\rv{H}}_{N+1} = i \mid \rv{H} = i ] \geq 1- \epsilon_N$. Notice that under the inference strategy, $\hat{\rv{H}}_{N+1} \neq i$ if and only if $\mathcal{C}_i(\vct{\rho}_{N+1}) < ND^*(i) - 2B\sqrt{N\log\frac{M}{\epsilon_N}} + \mathcal{C}_i(\vct{\rho}_{1})$. Therefore $\hat{\rv{H}}_{N+1} \neq i$, by Corollary \ref{cortriv}, implies that for some $j \neq i$
\begin{align}
     \sum_{n=1}^N\lambda_j^i(\rv{U}_n,\rv{Y}_n) <  ND^*(i) - 2B\sqrt{N\log\frac{M}{\epsilon_N}}.
\end{align}
Using the inequality established in (\ref{hoeff}) and a union bound, the probability of this event conditioned on hypothesis $i$ is at most $\epsilon_N$. Therefore, $\Py^{f,g}[\hat{\rv{H}}_{N+1} \neq i \mid \rv{H} = i ] \leq \epsilon_N$.
\end{proof}
\noindent Using Lemmas \ref{cslb2} and \ref{csub2} we can therefore conclude that
\begin{align}
    \lim_{N \to \infty}-\frac{1}{N}\log\phi^*_N(i) =  D^*(i).
\end{align}

\subsection{Proof of Theorem \ref{lbthm}}
In problem (\ref{opt1}), the strategies $f,g$ are required to satisfy the constraints $\psi_N(i) \leq  \epsilon_N$ for every $i \in \mathcal{H}$. For any pair $f,g$ of strategies that do satisfy all these constraints, we have the following because of Lemma \ref{cslb2}.
%\begin{align}
%    \gamma_N &\geq \sum_{i \in \mathcal{H}}\exp(-N(J^g_N(i)+ \mathcal{O}(\epsilon_N)))\\
%    &\geq \sum_{i \in \mathcal{H}}\exp(-N(D^*(i)+ \mathcal{O}(\epsilon_N))).
%\end{align}
\begin{align}
    \gamma_N  &= \sum_{i \in \mathcal{H}}(1-\rho_1(i))\phi_N(i)\\
    &\geq \sum_{i \in \mathcal{H}}(1-\rho_1(i))\exp(-NJ^g_N(i) -  K_1(i))\\
    &\geq \sum_{i \in \mathcal{H}}(1-\rho_1(i))\exp(-ND^*(i) -  K'_1(i)).
\end{align}
With $K_1$ defined as $K_1 := \max_{i \in \mathcal{H}}K'_1(i) \geq \max_{i \in \mathcal{H}}K_1(i)$, this proves Theorem \ref{lbthm}.
\begin{remark}\label{remark1}
We can obtain a tighter lower bound on $\gamma_N$ using Lemma \ref{cslb} without any restrictions on $\epsilon_N$, that is
\begin{align}\label{tightlb}
\gamma_N &\geq \sum_{i \in \mathcal{H}}(1-\rho_1(i))e^{\left(-NJ_{N}^g(i) - N\frac{2B\epsilon_N }{1-\epsilon_N}+\log(1-\epsilon_N)\right)}.
%&= \sum_{i \in \mathcal{H}}\rho_1(i)e^{\left(-NJ_{N}^g(i) -\mathcal{C}_i(\vct{\rho}_1) - N\frac{2B\epsilon_N }{1-\epsilon_N}+\log(1-\epsilon_N)\right)}.
\end{align}
\end{remark}

\subsection{Proof of Theorem \ref{achieve}}\label{proof2}
To prove Theorem \ref{achieve}, we construct appropriate experiment selection and inference strategies. We then show that these strategies achieve the desired bound on misclassification probability while satisfying the constraints in problem (\ref{opt1}). The construction is almost identical to the strategies in \cite{chernoff1959sequential}.

\subsubsection{Experiment selection strategy}Let the \emph{maximum a posteriori} (MAP) estimate at time $n$ be \begin{equation}\bar{i}_n := \argmax_{i \in \mathcal{H}}\rho_n(i).\end{equation} If $\bar{i}_n = i$, then an experiment is selected randomly with distribution $\vct{\alpha}^{i*}$. We denote this experiment selection strategy by $\bar{g}$.
\subsubsection{Inference strategy} Consider the strategy $\bar{f}$ where
$$
\bar{f}(\vct{\rho}_{N+1}) = 
\begin{cases}
i &\text{if } \mathcal{C}_i(\vct{\rho}_{N+1}) - \mathcal{C}_i(\vct{\rho}_{1}) \geq  ND^*(i) - N\delta\\
& \text{for some } i \in \mathcal{H},\\
\varnothing &\text{otherwise}.
\end{cases}
$$
Let us assume that $\delta < D^*(i)$ for every $i \in \mathcal{H}$ without loss of generality. This ensures that for large enough $N$, the threshold condition above can be satisfied by at most one hypothesis.

Using Lemma \ref{thresh}, we can conclude that for each hypothesis $i \in \mathcal{H}$, $\phi_N(i) \leq e^{-(ND^*(i) - N\delta)}$. Therefore, under these strategies, we have
\begin{align}
    \gamma_N \leq \sum_{i \in \mathcal{H}}(1-\rho_1(i))\exp(-N(D^*(i) -\delta)).
\end{align}
If we can show that the strategies also satisfy the type-$i$ error constraints in problem (\ref{opt1}), then clearly, $\gamma_N^* \leq \gamma_N$. To prove that these strategies do satisfy the constraints for large values of $N$, we use the arguments in \cite{chernoff1959sequential} and the details of this proof are as follows.

\begin{proof}
We now need to verify if the proposed strategy achieves the type-$i$ error constraints, that is $\psi_N(i) \leq \epsilon_N$ for each hypothesis $i \in \mathcal{H}$. Let us examine the evolution of the log likelihood ratio $\lambda_j^i$ associated with a pair of hypotheses $i$ and $j$ under the hypothesis $i$. Consider the following Doob decomposition
\begin{align}
    \sum_{n = 1}^N\lambda^i_j(\rv{U}_n,\rv{Y}_n) - ND^*(i) &=  \sum_{n = 1}^N\left[\lambda^i_j(\rv{U}_n,\rv{Y}_n) - \E_i[\lambda^i_j(\rv{U}_n,\rv{Y}_n) \mid \rv{I}_n] \right]
    + \sum_{n = 1}^N\left[\E_i[\lambda^i_j(\rv{U}_n,\rv{Y}_n) \mid \rv{I}_n] - D^*(i)\right]\\
    &=:  \sum_{n = 1}^N\rv{X}_n +  \sum_{n = 1}^N\rv{Z}_n.
\end{align}
Note that $\rv{X}_n$ is a martingale difference sequence with respect to the filtration $\rv{I}_n$ and $|\rv{X}_n| < 2B$ with probability 1. Using Azuma's inequality \cite{azuma1967weighted}, we have
\begin{align}
    \label{azuma1}\Py_i[\sum_{n = 1}^N\rv{X}_n < -K_3\sqrt{N\log\frac{2M}{\epsilon_N}}] \leq \exp\left(\frac{-K_3^2N\log\frac{2M}{\epsilon_N}}{8NB^2}\right).
\end{align}
We can choose $K_3 > 0$ such that
\begin{align}
    \exp\left(\frac{-K_3^2N\log\frac{2M}{\epsilon_N}}{8NB^2}\right) \leq \frac{\epsilon_N}{2M}.
\end{align}
Let $\rv{T}$ be the smallest time index such that $\bar{i}_n = i$ for every $n \geq \rv{T}$. Notice that $\rv{T}$ is a random variable. Under Assumption \ref{steadinf}, it was shown in \cite{chernoff1959sequential} (Lemma 1) that there exist constants $b,K > 0$ such that for every $i\in \mathcal{H}$, we have $\Py_i[\rv{T} > n] \leq Ke^{-bn}$.
Notice that
\begin{align}
    |\sum_{n = 1}^N\rv{Z}_n| < 2B\rv{T}.
\end{align}
Therefore, if
\begin{align}
    \sum_{n = 1}^N\rv{Z}_n &< -K_2\log\frac{2MK}{\epsilon_N}\\
    \implies \rv{T} &> \frac{K_2\log\frac{2MK}{\epsilon_N}}{2B}.
\end{align}
Therefore,
\begin{align}
    \label{azuma2}\Py_i[ \sum_{n = 1}^N\rv{Z}_n &< -K_2\log\frac{2MK}{\epsilon_N}] \leq Ke^{-\frac{bK_2\log\frac{2MK}{\epsilon_N}}{2B}}
\end{align}
We can pick a $K_2$ such that for large enough $N$
\begin{align}
    Ke^{-\frac{bK_2\log\frac{2MK}{\epsilon_N}}{2B}} \leq \frac{\epsilon_N}{2M}.
\end{align}
Using inequalities (\ref{azuma1}) and (\ref{azuma2}), we can conclude that
\begin{align}
    \Py_i[\bigcup_{j \neq i}\{\sum_{n = 1}^N\lambda^i_j(\rv{U}_n,\rv{Y}_n) < ND^*(i) -K_3\sqrt{N\log\frac{2M}{\epsilon_N}} - K_2 \log\frac{2MK}{\epsilon_N} \}] \leq \epsilon_N.
\end{align}
Because of Assumption \ref{epsassum}, for large enough $N$, we have
\begin{align}
ND^*(i) - K_3\sqrt{N\log\frac{2M}{\epsilon_N}}  - K_2 \log\frac{2MK}{\epsilon_N} > ND^*(i) - N\delta.
\end{align}
Using this fact and Corollary \ref{cortriv}, we have
\begin{align}
&\Py_i[\hat{\rv{H}}_{N+1} \neq i \mid \rv{H} = i]\\
&\leq \Py_i[\bigcup_{j \neq i}\{\sum_{n = 1}^N\lambda^i_j(\rv{U}_n,\rv{Y}_n) < ND^*(i) -N\delta \}]\\
& \leq \Py_i[\bigcup_{j \neq i}\{\sum_{n = 1}^N\lambda^i_j(\rv{U}_n,\rv{Y}_n) < ND^*(i) -K_3\sqrt{N\log\frac{2M}{\epsilon_N}} - K_2 \log\frac{2MK}{\epsilon_N} \}] \leq \epsilon_N.
\end{align}
Therefore, we can conclude that $\psi_N(i) \leq \epsilon_N$ for every $i \in \mathcal{H}$.
\end{proof}

\section{Discussion on Strategy Design}\label{discussion}
In Section \ref{proof2}, we described an experiment selection strategy $\bar{g}$. As discussed earlier, the agent starts with a prior belief on the set of hypotheses and as it performs experiments, its confidence on the true hypothesis improves. We refer to this initial phase of experimentation as the \emph{exploration phase}. Soon enough, the MAP estimate $\bar{i}_n = \rv{H}$. This implies that the agent starts selecting experiments using the distribution $\vct{\alpha}^{\rv{H}*}$ which rapidly improves its confidence on $\rv{H}$. We refer to this subsequent phase of experimentation as the \emph{verification phase}. 

According to Lemma 1 in \cite{chernoff1959sequential}, the exploration phase terminates in $\mathcal{O}(\log N)$ time with high probability under Assumption \ref{steadinf}. We can relax Assumption $\ref{steadinf}$ using the technique in \cite{nitinawarat2013controlled} and show that the exploration phase terminates in sublinear time with high probability. Therefore, in the asymptotic analysis, the impact of exploration on the overall performance is negligible. However, in the non-asymptotic regime, the exploration performance may have a significant impact on the overall performance, especially in problems like dynamic search over trees. This issue was discussed in \cite{naghshvar2013active} and \cite{wang2017active} in a stopping time setting and heuristic strategies were proposed to improve the exploration performance. One such heuristic is based on Extrinsic Jensen-Shannon (EJS) divergence \cite{naghshvar2012extrinsic}. Using our notation, the EJS divergence associated with an experiment $u$ and posterior belief $\vct{\rho}_n$ is the expected increment in confidence level on $\rv{H}$, that is
\begin{align}
EJS(\vct{\rho}_n,u) = \E[\mathcal{C}_{\rv{H}}(\vct{\rho}_{n+1}) - \mathcal{C}_{\rv{H}}(\vct{\rho}_n) \mid \vct{\rho}_n,\rv{U}_n = u].
\end{align}
The heuristic strategy in \cite{naghshvar2012extrinsic} is to greedily select the experiment that maximizes $EJS(\vct{\rho}_n,u)$ at time $n$. 

We described a lower bound on the error probability $\gamma_N$ in Remark \ref{remark1}. Using Jensen's inequality, we can further weaken the bound (\ref{tightlb}) to obtain the following lower bound
\begin{align*}
&\frac{1}{N}\log \frac{1}{\gamma_N}
%&\leq \sum_{i \in \mathcal{H}}\rho_1(i)(J_{N}^g(i) + \frac{1}{N}\mathcal{C}_i(\vct{\rho}_1)) + \frac{2B\epsilon_N }{1-\epsilon_N}+\frac{1}{N}\log(1-\epsilon_N)\\
\leq \frac{1}{N}\E^g[\mathcal{C}_{\rv{H}}(\vct{\rho}_{N+1})]  + \frac{2B\epsilon_N }{1-\epsilon_N}-\frac{1}{N}\log(1-\epsilon_N). 
\end{align*}
Therefore, as a heuristic, one can use the expected confidence level $\E^g[\mathcal{C}_{\rv{H}}(\vct{\rho}_{N+1})]$ as a proxy for $-\log\gamma_N$ and try to maximize the confidence level instead of the error rate. This is equivalent to maximizing $J_N^g$ where
\begin{align}
J_N^g := \E J_N^g(\rv{H}) = \frac{1}{N} \E^{g} \left[\mathcal{C}_{\rv{H}}(\vct{\rho}_{N+1})- \mathcal{C}_{\rv{H}}(\vct{\rho}_1) \right].
\end{align}
It was shown in \cite{kartik2018policy} that maximizing $J_N^g$ can be formulated as a Partially Observable Markov Decision Problem (POMDP). Using heuristics for solving POMDPs, we can approximately optimize $J_N^g$ and one such heuristic was presented in \cite{kartik2018policy}. The strategy based on EJS divergence \cite{naghshvar2012extrinsic} happens to be a one-step greedy policy with respect to this POMDP.

\section{Conclusions}\label{conc}
We formulated a fixed horizon active hypothesis testing problem in which the agent can decide on one of the hypotheses or declare its experiments inconclusive. For analyzing this problem, we formulated a sub-problem which is a generalization of Chernoff-Stein lemma \cite{cover2012elements} to a setting with multiple hypotheses and multiple experiments. We obtained lower bounds on optimal error probability in the sub-problem and used them to obtain lower bounds on misclassification probability in our original problem. We also derived upper bounds by constructing appropriate strategies and analyzing their performance. We defined a quantity called expected confidence rate and based on it, we proposed a heuristic approach for strategy design.

\section*{Acknowledgments}
This research was supported, in part, by National Science Foundation under Grant NSF CNS-1213128, CCF-1410009, CPS-1446901, Grant ONR N00014-15-1-2550, and Grant AFOSR FA9550-12-1-0215.

%%%%%%
%% To balance the columns at the last page of the paper use this
%% command:
%%
%\enlargethispage{-1.2cm} 
%%
%% If the balancing should occur in the middle of the references, use
%% the following trigger:
%%
%\IEEEtriggeratref{3}
%%
%% which triggers a \newpage (i.e., new column) just before the given
%% reference number. Note that you need to adapt this if you modify
%% the paper.  The "triggered" command can be changed if desired:
%%
%\IEEEtriggercmd{\enlargethispage{-20cm}}
%%
%%%%%%

%%%%%%
%% References:
%% We recommend the usage of BibTeX:
%%
%\bibliographystyle{IEEEtran}
%\bibliography{definitions,bibliofile}
%%
%% where we here have assume the existence of the files
%% definitions.bib and bibliofile.bib.
%% BibTeX documentation can be obtained at:
%% http://www.ctan.org/tex-archive/biblio/bibtex/contrib/doc/
%%%%%%

%% Or you use manual references (pay attention to consistency and the
%% formatting style!):
\bibliographystyle{IEEEtran}
\bibliography{refs1}

\end{document}